\documentclass[12pt]{amsart}
\usepackage{hyperref,breakurl}
\usepackage{versions}\excludeversion{ignore}
\usepackage{amssymb,amsthm,latexsym,amsmath,mathdots}
\usepackage{epstopdf}
\usepackage[dvipsnames,svgnames]{pstricks}
\newcommand{\cA}{DodgerBlue}
\newcommand{\cB}{Black}
\newcommand{\cC}{Crimson}
\newcommand{\A}{{\color{\cA}A}}
\newcommand{\B}{{\color{\cB}B}}
\newcommand{\C}{{\color{\cC}C}}
\newtheorem{theorem}{Theorem}
\newtheorem{corollary}{Corollary}
\let\rel=\mathrel
\newcommand{\stack}[1]{{\def\\{|}{\{#1\}}}}
\newcommand{\rstack}[1]{\rel{\stack{#1}}}
\title{Tripartite Unions}
\author{Nachum Dershowitz}
\date{\today}                                           % Activate to display a given date or no date

\begin{document}
\maketitle
\begin{abstract}
This note provides conditions under which the union of three well-founded binary relations is also well-founded.
\end{abstract}

This note concerns conditions under which the union of several well-founded (binary) relations is also well-founded.%
\footnote{By \emph{well-founded}, we mean the absence of infinite forward-pointing paths.
For some of the history of well-foundedness based on Ramsey's Theorem, 
see Pierre Lescanne's \textit{Rewriting List}, contributions 38--41 at \url{http://www.ens-lyon.fr/LIP/REWRITING/CONTRIBUTIONS}
and
Andreas Blass and Yuri Gurevich,  ``Program Termination and Well Partial Orderings'', \textit{ACM Transactions on Computational Logic} \textbf{9}(3), 2008 (available at \url{http://research.microsoft.com/en-us/um/people/gurevich/Opera/178.pdf}).}

To garner insight, we tackle
just three relations, $A$, $B$, and $C$, over some underlying set $V$.
Let $$\stack{A\\ B}$$ denote $A\cup B$, and so on for other unions of relations.
And let juxtaposition indicate composition of relations
and superscript $*$ signify transitive closure.
We'll refer to the relations as ``colors''.

\begin{theorem}[Ramsey]
The union $\stack{A\\B\\C}$ is well-founded if
\begin{eqnarray}
\stack{A\\B\\C} \stack{A\\B\\C} & \subseteq & \stack{A\\B\\C}
\end{eqnarray}
\end{theorem}

\begin{proof}
The infinite version of Ramsey's Theorem applies when the union is transitive, so that every two (distinct) nodes
within an infinite chain in the union of the colors has a colored
(directed) edge.
Then,  there must lie an infinite monochrome subchain within any infinite chain, 
contradicting the well-foundedness of each color alone.%
\footnote{See
Alfons Geser, \textit{Relative Termination}, Ph.D.\@ dissertation, Fakult\"at fu\"r Mathematik und Informatik, Universit\"at Passau, Germany, 1990 (Report 91-03, Ulmer Informatik-Berichte, Universit\"at Ulm, 1991; available at \url{http://homepage.cs.uiowa.edu/~astump/papers/geser_ dissertation.pdf}).}
\end{proof}

Only three of the nine cases are actually needed for the limited outcome that we are seeking
(an infinite monochromic path, rather than a clique---as in Ramsey's Theorem), as we observe next.

\begin{theorem}
The union $\stack{A\\B\\C}$ is well-founded if
\begin{eqnarray}
B A \cup CA \cup CB & \subseteq & \stack{A\\B\\C} \;.
\end{eqnarray}
\end{theorem}

\begin{proof}
When the union is not well-founded, there is an infinite path $X=\{x_i\}_i$ with each edge from $x_i$ to $X_{i+1}$ one of $A$, $B$, or $C$.
Extract a maximal subsequence  $\{x_{i_j}\}_j$ of $X$ such that $x_{i_j} \rel{A} x_{i_{j+1}}$ for each $j$.
If it's finite, then repeat at the first opportunity in the tail.
If any is infinite, we have our contradiction.
If they're all finite, then consider the first occurrence of $x \rstack{B\\C} y \rel{A} z$.
Since we could not take an $A$-step from $x$, or we would have, the  conditions tell us that $x \rstack{B\\C} z$.
Swallowing up all such (non-initial) $A$-steps in this way, we are left with an infinite chain in $B\cup C$,
for which we also know that no $A$-steps are possible anywhere.
Now extract maximal $B$-chains and then erase them, replacing $x \rel{C} y \rel{B} z$
with $x \rel{C} z$ ($A$- and $B$-steps having been precluded), leaving an infinite chain colored purely $C$.
\end{proof}

\begin{corollary}
If $A$, $B$, and $C$ are transitive and
\begin{eqnarray*}
B A \cup CA \cup CB & \subseteq & \stack{A\\B\\C} \;,
\end{eqnarray*}
then, whenever there is an infinite path in the union $\stack{A\\B\\C}$, there is an infinite monochromatic clique.
\end{corollary}

We can do considerably better than the previous theorem:

\begin{theorem}[Tripartite]
The union $\stack{A\\B\\C}$ is well-founded if
\begin{eqnarray}
\stack{B\\C} A  & \subseteq & A\stack{A\\B\\C}^* \cup B \cup C\\\nonumber
CB  & \subseteq & A\stack{A\\B\\C}^* \cup BB^* \cup C \;.
\end{eqnarray}
\end{theorem}

Let's call the existence of an infinite outgoing chain in the union $\stack{A\\B\\C}$ \emph{immortality}.

\begin{proof}
We first construct an infinite chain $X=\{x_i\}_i$, in which an $A$-step is always preferred over $B$ or $C$, as long as immortality is maintained.
To do this, we start with an immortal element $x_0\in V$.
At each stage in the construction, if the chain so far ends in $x_i$, we look to see if there is any $y$ such that $x_i \rel{A} y$ and from which proceeds some infinite chain in the union, in which case $y$ is chosen to be $x_{i+1}$.
Otherwise, $x_{i+1}$ is any immortal element $z$ such that $x_i \rel{B} z$ or $x_i \rel{C} z$.

If there are infinitely many $B$'s and/or $C$'s in $X$, use them---by means of the first condition---to remove all subsequent $A$-steps, leaving only $B$- and $C$-steps going out of points from which $A$ leads of necessity to mortality.
From what remains, if there is any $C$-step at a point where one could take one or more $B$-steps to anyplace later in the chain, take the latter route instead.
What remains now are $C$-steps at points where $BB^*$ detours are also precluded.
If there are infinitely many such $C$-steps,
then applying the  condition for $CB$ will result in a pure $C$-chain, because neither $A\stack{A\\B\\C}^* $ nor $ BB^*$ are options.
\end{proof}

Dropping $C$ from the conditions of the previous theorem, one gets
the \emph{jumping} criterion for well-foundedness of the union of two well-founded relations $A$ and $B$:%
\footnote{See Henk Doornbos and Burghard von Karger, ``On the Union of Well-Founded Relations'', \textit{Logic Journal of the IGPL} \textbf{6}(2), pp.\@ 195--201, 1998 (available at \url{http://citeseerx.ist.psu.edu/viewdoc/download?doi=10.1.1.28.8953&rep=rep1&type=pdf}). The property is called ``jumping'' in
Nachum Dershowitz, ``Jumping and Escaping: Modular Termination and the Abstract Path Ordering'', \textit{Theoretical Computer Science} \textbf{464}, pp.\@ 35--47, 2012 (available at \url{http://nachum.org/papers/Toyama.pdf}).}
\begin{eqnarray*}
B A & \subseteq & A\stack{A\\B}^* \cup B \;.
\end{eqnarray*}
Applying this criterion twice,
one gets somewhat different (incomparable) conditions for well-foundedness.
 
\begin{theorem}[Jumping]
The union $\stack{A\\B\\C}$ is well-founded if
\begin{eqnarray}
B A & \subseteq & A\stack{A\\B}^* \cup B\\\nonumber
C \stack{A\\B} & \subseteq & \stack{A\\B}\stack{A\\B\\C}^* \cup C \;.
\end{eqnarray}
\end{theorem}

\begin{proof}
The first inequality is the jumping criterion.
The second is the same with $C$ for $B$ and $\stack{A\\B}$ in place of $A$.
\end{proof}

For two relations, jumping provides a substantially weaker criterion for well-foundedness than does the appeal to Ramsey.
But for three,
whereas jumping allows more than one step for $BA$ (in essence, $A A^* B^*$),
it doesn't allow for $C$, which Ramsey does.

Switching r\^oles, start with jumping for $\stack{B\\C}$ before combining with $A$,
we get slightly different conditions yet:

\begin{theorem}[Jumping]
The union $\stack{A\\B\\C}$ is well-founded if
\begin{eqnarray}
CB & \subseteq & B\stack{B\\C}^* \cup C \\\nonumber
\stack{B\\C}A & \subseteq & A\stack{A\\B\\C}^* \cup B \cup C\;.
\end{eqnarray}
\end{theorem}

Both this version of jumping and our tripartite condition allow
\begin{eqnarray*}
\stack{B\\C}A & \subseteq & A\stack{A\\B\\C}^* \cup B \cup C\\
CB & \subseteq & BB^* \cup C \;.
\end{eqnarray*}
They differ in that jumping also allows
\begin{eqnarray*}
CB & \subseteq & B\stack{B\\C}^* \;,
\end{eqnarray*}
whereas tripartite has
\begin{eqnarray*}
CB & \subseteq & A\stack{A\\B\\C}^* 
\end{eqnarray*}
instead.

Sadly, we cannot have the best of both worlds.
Let's color edges $\A$, $\B$, and $\C$ with
(solid) {\color{DodgerBlue}azure}, (dotted) black, and (dashed) {\color{Crimson}crimson} ink, respectively.
The graph
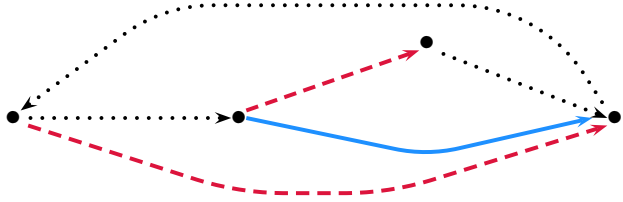
\begin{center}
\begin{pspicture}(-4,-1.5)(4,2)
\psline[linewidth=1.5pt,linecolor=Crimson,linestyle=dashed,linearc=4]{->}(-3.8,-0.1)(-1,-1.0)(1,-1)(3.9,-0.1) % v->z
\psline[linewidth=1.5pt,linecolor=Black,linestyle=dotted]{->}(-3.9,0)(-1.1,0) % v->x
\psline[linewidth=1.5pt,linecolor=DodgerBlue,,linearc=2]{->}(-0.9,0)(1.5,-0.5)(3.7,0) % x->z
\psline[linewidth=1.5pt,linecolor=Black,linestyle=dotted,linearc=2]{->}(3.9,0.1)(3,1.5)(-2,1.5)(-3.9,0.1) % z->v
\psline[linewidth=1.5pt,linecolor=Crimson,linestyle=dashed,linearc=2]{->}(-0.9,0.1)(1.4,0.9) % x->y
\psline[linewidth=1.5pt,linecolor=Black,linestyle=dotted,linearc=2]{->}(1.6,0.9)(3.9,0) % y->z
\rput(-4,0){$\bullet$}\rput(-1,0){$\bullet$}\rput(1.5,1){$\bullet$}\rput(4,0){$\bullet$} % vxyz
\end{pspicture}
\end{center}
only has multicolored loops
despite satisfying
\begin{equation*}
\begin{array}{rcl}
\stack{\B\\\C}\A & \subseteq & \C\\
\C \B & \subseteq & \A \cup \B\stack{\B\\\C}^* \;.
\end{array}
\end{equation*}
Even
\begin{equation*}
\begin{array}{rcl}
\stack{\B\\\C}\A & \subseteq & \C\\
\C \B & \subseteq & \B\stack{\A\\\B}^*
\end{array}
\end{equation*}
doesn't work. To wit, the double loop in 
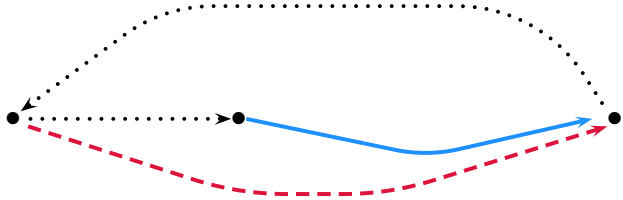
\begin{center}
\begin{pspicture}(-4,-1.5)(4,2)
\psline[linewidth=1.5pt,linecolor=Crimson,linestyle=dashed,linearc=4]{->}(-3.8,-0.1)(-1,-1.0)(1,-1)(3.9,-0.1) % v->z
\psline[linewidth=1.5pt,linecolor=Black,linestyle=dotted]{->}(-3.9,0)(-1.1,0) % v->x
\psline[linewidth=1.5pt,linecolor=DodgerBlue,,linearc=2]{->}(-0.9,0)(1.5,-0.5)(3.7,0) % x->z
\psline[linewidth=1.5pt,linecolor=Black,linestyle=dotted,linearc=2]{->}(3.9,0.1)(3,1.5)(-2,1.5)(-3.9,0.1) % z->v
\rput(-4,0){$\bullet$}\rput(-1,0){$\bullet$}\rput(4,0){$\bullet$} % vxz
\end{pspicture}
\end{center}
harbors no monochrome subchain.
By the same token, 
\begin{center}
\begin{pspicture}(-4,-1.5)(4,2)
\psline[linewidth=1.5pt,linecolor=Crimson,linestyle=dashed,linearc=4]{->}(-3.8,-0.1)(-1,-1.0)(1,-1)(3.9,-0.1) % v->z
\psline[linewidth=1.5pt,linecolor=Black,linestyle=dotted]{->}(-3.9,0)(-1.1,0) % v->x
\psline[linewidth=1.5pt,linecolor=DodgerBlue,,linearc=2]{->}(-0.9,0)(1.5,-0.5)(3.7,0) % x->z
\psline[linewidth=1.5pt,linecolor=DodgerBlue,linearc=2]{->}(3.9,0.1)(3,1.5)(-2,1.5)(-3.9,0.1) % z->v
\rput(-4,0){$\bullet$}\rput(-1,0){$\bullet$}\rput(4,0){$\bullet$} % vxz
\end{pspicture}
\end{center}
counters the putative hypothesis
\begin{equation*}
\begin{array}{rcl}
\B \A  \cup \C \B & \subseteq & \C\\
\C \A & \subseteq & \B \A^* \;.
\end{array}
\end{equation*}

% General case $n$
% Consider the case where all or some are transitive.

\end{document}